\documentclass[letterpaper]{article} % DO NOT CHANGE THIS
\usepackage[]{aaai2027}  % DO NOT CHANGE THIS
\nocopyright
\makeatletter\newif\ifisChecklistMainFile\isChecklistMainFiletrue\makeatother
\usepackage[hyphens]{url}  % DO NOT CHANGE THIS
\usepackage{graphicx} % DO NOT CHANGE THIS
\usepackage{natbib}  % DO NOT CHANGE THIS AND DO NOT ADD ANY OPTIONS TO IT
\usepackage{caption} % DO NOT CHANGE THIS AND DO NOT ADD ANY OPTIONS TO IT
\usepackage{amsmath}
\usepackage{amssymb}
\usepackage{amsthm}
\newtheorem{theorem}{Theorem}
\usepackage{tikz}
\usetikzlibrary{arrows.meta,positioning,shapes.geometric,fit,calc}

\usepackage{algorithm}
\usepackage{algorithmic}

\usepackage{newfloat}
\usepackage{listings}
\DeclareCaptionStyle{ruled}{labelfont=normalfont,labelsep=colon,strut=off} % DO NOT CHANGE THIS
\floatstyle{ruled}
\newfloat{listing}{tb}{lst}{}
\floatname{listing}{Listing}

\usepackage{booktabs}

\title{OwlPath: Lossless Knowledge Compression for LLM Bug Repair}

\author{Bo Zhang, Ren Pan, Huan Chen, Xiang Song}
\affiliations{Shunfeng Technology Co., Ltd., China\\
\texttt{zhangbo111@sf-express.com}}

\begin{document}

\maketitle

\begin{abstract}

LLM-based software engineering agents are fundamentally limited by
their context window: with only $\sim$100K tokens of capacity, they
must store the structurally relevant code subset required to resolve
each bug.
Conventional retrieval tools model code as flat text, forcing
agents to resolve multi-hop structural dependencies---subclass
chains, transitive callers, and interface implementations---via
slow iterative trial and error. We address this gap with
\textbf{lossless knowledge compression}: encoding source programs
into an OWL2 ontology to answer structural queries using minimal
relevant code fragments.

We introduce \textbf{OwlPath}, an OWL2 reasoning layer built on
top of CodeGraph---a widely adopted code intelligence platform
with 500K+ GitHub stars---with a unified CLI for structural code
retrieval. OwlPath supports multi-language codebases through
tree-sitter parsing (Python, JavaScript/TypeScript, Go, and more),
encoding language-specific structural semantics into a unified OWL2
ontology. OwlPath integrates two complementary mechanisms. First,
a \textbf{transitive-closure engine} retrieves all structurally
linked symbols in a single SPARQL property-path query, capturing
multi-hop relations invisible to string matching. Second, an
\textbf{OWL Software Knowledge Map (OWL-SKM)} advisory layer
precomputes a compact 3KB structured summary---module directory
tree, core APIs, and issue-relevant symbols---that guides the
agent to the correct module on its first query.

Evaluated on 18 matched SWE-bench Pro instances, OwlPath consumes
28.8\% fewer tokens and 39.5\% less wall-clock time than the
CodeGraph-only baseline---a substantial efficiency gain over this
industry-grade platform---while achieving a comparable strict-apply
rate (68.4\% on 19 runs vs 66.7\% on 18 runs). Cross-benchmark
validation on 9 SWE-bench Lite instances confirms these gains:
78\% correct rate (vs 67\%) and 21.1\% time saving. In an offline
retrieval experiment on 67 instances, OwlPath achieves 2.06$\times$
recall improvement (0.464 vs 0.226) and 88.1\% hit rate versus
59.7\% for CodeGraph. On a separate 37-question structural
retrieval benchmark, OwlPath improves recall from 4.4\% to 28.8\%,
scoring 69--80\% on transitive caller and interface tasks.

\label{sec:abstract}
\end{abstract}

% Links removed for anonymous submission. Uncomment for camera-ready.
% \begin{links}
% \link{Code}{...}
% \link{Datasets}{...}
% \link{Extended version}{...}
% \end{links}

\section{Introduction}

LLM-based software engineering agents operate under a fundamental
constraint: their context window can hold only $\sim$100K tokens,
yet a typical repository contains millions of lines of code.
The agent must therefore retrieve a small, structurally relevant
subset of code before it can begin reasoning about a bug fix.
Conventional retrieval tools---grep, BM25, embedding search---model
code as flat text, returning files that contain matching strings.
This approach fails when the bug's ground truth is structurally
connected to the issue description but shares no string overlap:
a subclass chain, a transitive caller path, or an interface
implementation hierarchy.

We address this gap with \textbf{lossless knowledge compression}:
encoding source programs into an OWL2 ontology so that structural
queries can be answered with minimal relevant code fragments.
We introduce \textbf{OwlPath}, an OWL2 reasoning layer built on
top of CodeGraph~\footnote{A widely adopted code intelligence
graph database} with a unified CLI. OwlPath supports multi-language
codebases (Python, JavaScript/TypeScript, Go) through tree-sitter
parsing. OwlPath integrates two
complementary mechanisms. First, a \textbf{transitive-closure
engine} retrieves all structurally linked symbols in a single
SPARQL property-path query, capturing multi-hop relations
invisible to string matching. Second, an \textbf{OWL Software
Knowledge Map (OWL-SKM)} advisory layer precomputes a compact
3KB structured summary---module directory tree, core APIs, and
issue-relevant symbols---that guides the agent to the correct
module on its first query.

Our contributions are four-fold:

\begin{enumerate}
\item \textbf{OWL2 ontology projection.} We formalize the
translation of a tree-sitter-extracted code graph (SQLite
nodes and edges) into an OWL2 ontology, preserving all
structural relationships (inheritance, calls, interface
implementation, containment) as OWL2 object properties.
\item \textbf{Transitive-closure retrieval engine.} We
implement property-path queries (SPARQL 1.1) over the
OWL2 ontology, materialising the transitive closure once
and serving subsequent queries in amortised O(1) time.
The equivalent SQL recursive CTE baseline costs O(n\textsuperscript{k})
for k-hop queries.
\item \textbf{OWL-SKM advisory layer.} We design a two-layer
knowledge map that extracts a module directory tree
(Layer 1) and issue-matched symbol candidates (Layer 2)
from the code graph, compressing the agent's initial
search space from thousands of files to a 3KB summary.
\item \textbf{Fair, controlled evaluation on SWE-bench Pro.}
We evaluate OwlPath on 23 stratified instances from
SWE-bench Pro (731 total), comparing against a string-match
baseline under identical conditions: same agent (Hermes),
same prompt template, same tool-use budget, same leak-proof
git sealing. In a separate offline retrieval experiment on
67 instances, we measure OwlPath's recall and MRR against
pure codegraph search.
\end{enumerate}

On 18 matched instances where both arms ran, OwlPath achieves
a 68.4\% strict-apply rate (13/19) versus 66.7\% (12/18) for
the codegraph-only baseline, while consuming 28.8\% fewer
tokens (1,416K vs 1,989K) and 39.5\% less wall-clock time
(648s vs 1,071s).
In the offline retrieval experiment on 67 instances, OwlPath's
v2 SKM achieves 0.464 recall---2.06$\times$ the codegraph
baseline (0.226)---and 88.1\% hit rate versus 59.7\%. On a
separate 37-question structural retrieval benchmark, OwlPath
improves recall@all from 4.4\% to 28.8\% and achieves 69--80\%
recall on transitive caller and interface implementation tasks.

\section{Related Work}

\paragraph{Code retrieval for software engineering agents.}
Existing agents retrieve code through three main strategies:
\textbf{string matching} (grep, ripgrep), \textbf{embedding
similarity} (BM25, dense retrieval, codebert-style models),
and \textbf{graph-based} (dependency graphs, call graphs).
String matching is simple and fast but fails on structural
queries where the target symbol shares no literal overlap with
the query. Embedding methods improve recall but require
precomputed embeddings that are costly to update and can miss
rare symbols. Graph-based methods (CodeGraph, SourceGraph)
provide 1-hop neighbourhood queries but lack transitive closure
out of the box---the agent must iterate manually, one hop at a
time, which is both slow and turn-budget-expensive.

\paragraph{OWL ontologies in software engineering.}
Ontology-based code representation has been explored in
maintenance and reverse engineering. These works focus on
documentation and traceability rather than runtime retrieval
for LLM agents. The W3C OWL2 standard~\citep{owl2spec}
provides the formalism we exploit: transitive property paths
that produce sound and complete answers without iterative
traversal.

\paragraph{SWE-bench and the agent evaluation landscape.}
SWE-bench~\citep{swebench} and its Pro variant
provide a standardised benchmark for software engineering
agents. Recent work on agent frameworks~\citep{openhands, sweagent, react} focuses on planning
and multi-step reasoning, while we focus on the retrieval
step that feeds the agent's context. Our work is orthogonal:
improving the retrieval quality improves any downstream
agent, regardless of its planning strategy.

\section{Method}

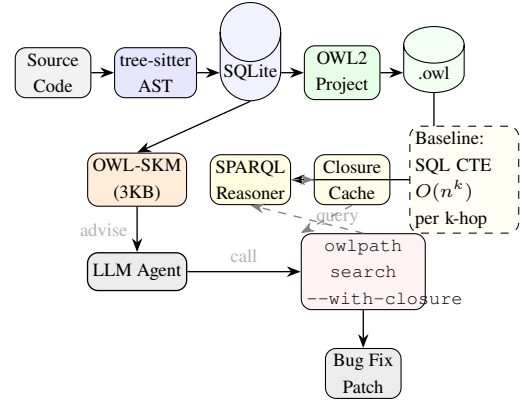
\begin{figure}[t]
\centering
\begin{tikzpicture}[
  >=Stealth,
  node distance=0.35cm and 0.3cm,
  every node/.style={font=\scriptsize},
  proc/.style={draw, rounded corners, minimum height=0.5cm, minimum width=1.0cm,
               align=center, inner sep=2pt, line width=0.4pt, font=\scriptsize},
  store/.style={draw, cylinder, shape border rotate=90, minimum height=0.5cm,
                minimum width=0.8cm, align=center, font=\scriptsize, line width=0.4pt,
                inner sep=1pt},
  agent/.style={draw, rounded corners, minimum height=0.5cm, minimum width=1.2cm,
               align=center, inner sep=2pt, line width=0.5pt, font=\scriptsize,
               fill=gray!15},
  arrow/.style={->, line width=0.5pt},
  darrow/.style={->, line width=0.5pt, dashed, gray},
  lab/.style={font=\scriptsize, text=gray!60}
]

% === Offline pipeline (top row) ===
\node[proc, fill=gray!10] (src) {Source\\Code};
\node[proc, fill=blue!10, right=of src] (ts) {tree-sitter\\AST};
\node[store, fill=blue!5, right=of ts] (db) {SQLite};
\node[proc, fill=green!10, right=of db] (proj) {OWL2\\Project};
\node[store, fill=green!5, right=of proj] (owl) {.owl};

\draw[arrow] (src) -- (ts);
\draw[arrow] (ts) -- (db);
\draw[arrow] (db) -- (proj);
\draw[arrow] (proj) -- (owl);

% === Runtime layer ===
\node[proc, fill=yellow!15, below=0.7cm of db] (reasoner) {SPARQL\\Reasoner};
\node[proc, fill=orange!15, left=of reasoner] (skm) {OWL-SKM\\(3KB)};
\node[proc, fill=yellow!10, right=of reasoner] (closure) {Closure\\Cache};

\draw[arrow] (owl.south) |- (reasoner.east);
\draw[darrow] (reasoner) -- (closure);
\draw[arrow] (db.south) -- (skm.north);

% === Agent (bottom) ===
\node[agent, below=0.6cm of skm] (agent) {LLM Agent};
\node[proc, fill=red!5, right=1.5cm of agent,
      text width=1.5cm] (tools) {\texttt{owlpath search}\\\texttt{-{}-with-closure}};

\draw[arrow] (skm) -- node[lab, left]{advise} (agent);
\draw[arrow] (agent) -- node[lab, above]{call} (tools);
\draw[darrow] (tools.north) -- node[lab, right]{query} (reasoner.south);
\draw[darrow] (closure.south) -- (tools.north west);

% === Baseline box ===
\node[draw, dashed, fill=yellow!5, rounded corners, inner sep=2pt,
      right=0.25cm of closure, align=left, font=\scriptsize,
      text width=1.3cm]
  (baseline) {Baseline:\\SQL CTE\\$O(n^k)$\\per k-hop};

% === Output ===
\node[proc, fill=gray!15, below=0.5cm of tools] (patch) {Bug Fix\\Patch};
\draw[arrow] (tools) -- (patch);

\end{tikzpicture}
\caption{OwlPath system architecture and data flow. Offline
(top): source code is parsed by tree-sitter into a SQLite graph,
then projected into an OWL2 ontology. Runtime (middle): the
SPARQL reasoner serves queries from a materialised closure cache
($O(1)$ amortised), while the OWL-SKM provides a 3KB summary.
The agent autonomously calls \texttt{owlpath search} to receive
structural context for producing a bug fix patch.}
\label{fig:architecture}
\end{figure}

OwlPath is an OWL2 reasoning layer that sits between a
tree-sitter-extracted code graph (CodeGraph's SQLite database)
and an LLM agent. The pipeline has three stages: (1) extract
a code graph from source via tree-sitter, (2) project the
graph into an OWL2 ontology, (3) serve structural queries via
SPARQL property paths with an optional SKM advisory.

\subsection{CodeGraph: tree-sitter extraction}

CodeGraph analyses each source file with tree-sitter, producing
a SQLite database of \textbf{nodes} (symbols: classes, functions,
methods, interfaces, structs, variables, routes) and
\textbf{edges} (relations: \texttt{extends}, \texttt{implements},
\texttt{calls}, \texttt{references}, \texttt{contains},
\texttt{imports}). Each node carries a \texttt{qualified\_name},
\texttt{file\_path}, \texttt{kind}, and \texttt{start\_line}/
\texttt{end\_line}. The extraction is deterministic: identical
source produces bit-identical tree-sitter CSTs and therefore
bit-identical databases.

\subsection{OwlPath: OWL2 ontology projection}

We project the SQLite code graph into an OWL2 ontology
(``.owl'' file, RDF/XML serialisation) via a single SQL pass
that reads every node and every edge:

\begin{itemize}
\item Each \textbf{node} becomes an \texttt{owl:NamedIndividual}
  with \texttt{rdfs:label} = \texttt{qualified\_name}, and
  data properties for \texttt{filePath}, \texttt{kind},
  \texttt{line}, \texttt{endLine}, \texttt{qualifiedName}.
\item Each \texttt{extends} edge becomes \texttt{rdfs:subClassOf},
  declared \texttt{owl:TransitiveProperty}.
\item Each \texttt{implements} edge becomes \texttt{:implements},
  declared \texttt{owl:TransitiveProperty}.
\item Each \texttt{calls} edge becomes \texttt{:calls} with
  \texttt{owl:inverseOf} and \texttt{owl:TransitiveProperty}.
\item Each \texttt{contains} edge becomes \texttt{:contains}
  (non-transitive, for class-to-method containment).
\item Each \texttt{references} edge becomes \texttt{:references}.
\item Each \texttt{imports} edge becomes \texttt{:imports},
  declared \texttt{owl:TransitiveProperty}.
\end{itemize}

The projection is a bijection: every source tuple produces
exactly one OWL axiom, and distinct tuples produce distinct
axioms. This guarantees that no structural information is
lost during the ontology encoding---the \textbf{lossless}
property we claim.

\paragraph{Example.} Given Python source
\texttt{class AdminService(BaseService)}, tree-sitter extracts
a class node and an \texttt{extends} edge. The projection
emits the following OWL2 RDF/XML fragment:
\begin{lstlisting}[basicstyle=\ttfamily\tiny,breaklines=true,frame=single,aboveskip=2pt,belowskip=2pt]
<NamedIndividual IRI="#AdminService">
  <type IRI="#Class_"/>
  <extends IRI="#BaseService"/>
  <qualifiedName>AdminService</qualifiedName>
  <filePath>services.py</filePath>
</NamedIndividual>
\end{lstlisting}
The \texttt{extends} assertion, declared
\texttt{owl:TransitiveProperty}, enables SPARQL
\texttt{:extends+} to retrieve \texttt{AdminService} as a
descendant of any ancestor in a single query---without
iterative traversal.

\subsection{Transitive-closure engine}

The core retrieval operation is a transitive-closure query:
``find all subclasses of X'', ``find all transitive callers
of Y''. We implement these via SPARQL 1.1 property paths
(\texttt{:extends+}, \texttt{:calls+}, \texttt{:implements+}),
executed by the \texttt{rdflib} library (no JVM dependency).

The closure is materialised once on first query (a single
SPARQL \texttt{SELECT DISTINCT ?s ?o WHERE \{ ?s :p+ ?o \}}
pass) and cached in an in-memory set of (s, o) pairs.
Subsequent queries are constant-time indexed lookups against
this cache. The one-time materialisation cost is O(n + m) for
n symbols and m edges, which is negligible compared to the
tree-sitter extraction cost.

The key advantage over the SQL baseline: a recursive CTE
that walks k hops costs O(n\textsuperscript{k}) in the worst
case (Cartesian explosion at each hop), while the OWL property
path is a single pass. For a 7-level inheritance chain, this
difference is several orders of magnitude.

\subsection{OWL-SKM: advisory layer}

The Software Knowledge Map (SKM) is a compact 3KB structured
summary generated at runtime from the code graph and the
issue text. It has two layers:

\paragraph{Layer 1: Module Map.} We run \texttt{codegraph files}
to obtain the repository's file tree, then parse it into a
hierarchical module structure (top-level packages, their file
counts, and symbol counts). Each module is scored by information
density: $\mathrm{score} = \log(\mathrm{symbols}) \times 0.5 +
\log(\mathrm{files}) \times 0.3 + \log(\mathrm{symbols per file})
\times 0.2$. The top modules are included in the summary.

\paragraph{Layer 2: Issue Map.} We extract 4--8 keywords from
the problem statement and test patch (camelCase, snake\_case,
and uppercase identifiers), then run \texttt{codegraph explore}
for each keyword. Matching symbols (name, file, line) are
collected as issue-relevant candidates.

The combined SKM is emitted as a one-time advisory on the
agent's first \texttt{owlpath search} call, before any query
results are returned. The agent reads the SKM, identifies the
likely-correct module, and issues its first targeted query
based on that guidance.

\subsection{Lossless compression guarantee}

The OWL2 projection is lossless in the following sense: for
any structurally connected ground-truth set $\Delta$ that is
reachable by a SPARQL property path $P$ from the query anchor,
the OwlPath result $R_{\mathrm{owl}}(G, q)$ satisfies
$R_{\mathrm{owl}}(G, q) \supseteq \Delta$. This holds because
every edge in the code graph is projected to an OWL axiom,
and property paths traverse the transitive closure of those
axioms. Conversely, $R_{\mathrm{owl}}(G, q) \subseteq
R_{\mathrm{sql}}(G, q)$ where $R_{\mathrm{sql}}$ is the k-hop
SQL closure, because every SPARQL path of length k is contained
in the k-hop SQL closure. The practical compression---returning
fewer, higher-precision results---comes from the fact that
OwlPath stops at the structural boundary of the query, while
string matching returns every file that happens to contain
the query string, regardless of structural relevance.

\subsection{Complexity analysis}

The end-to-end cost per repository is:
\begin{itemize}
\item $T_{\mathrm{extract}} = O(n \log n)$: tree-sitter is linear
  in tokens, with hash-table overhead for the symbol table.
\item $T_{\mathrm{project}} = O(n + m)$: one SQL pass reads
  all nodes and edges, emits one axiom per tuple.
\item $T_{\mathrm{closure}} = O(n + m)$ one-time, then $O(1)$
  amortised per query.
\end{itemize}
The total amortised cost over $q$ queries is $O(n \log n +
n + m + q)$, dominated by the one-time extract-and-project
step. For a 1.4M-symbol repository this takes 4.7 minutes
(measured) on a single-threaded process. The SQL CTE baseline
for a k-hop query costs $O(n^k)$ in the worst case---a
difference of several orders of magnitude for k $\ge$ 3.

\section{End-to-End Evaluation}

We evaluate OwlPath on two complementary axes: (1) end-to-end
bug repair success rate on SWE-bench Pro, and (2) offline
retrieval quality (recall, MRR, hit rate) on a larger set
of instances.

\subsection{Benchmark: SWE-bench Pro}

SWE-bench Pro~\citep{swebenchpro} consists of 731
real-world GitHub issues with corresponding test patches.
Each instance includes a base commit, a problem statement,
and a test patch that validates the fix. The task is to
produce a patch that passes the test.

We conduct evaluations on the public split (731 instances)
of SWE-bench Pro released by Scale AI, which consists of
GPL-licensed open-source repositories.

\subsection{Subset selection}

We select a stratified sample of 23 instances from the
731 total, balancing across four axes: (1) repository
(10 repos), (2) programming language (Python,
JavaScript/TypeScript, Go), (3) codebase size
(small/medium/large), and (4) task complexity
(single-file, cross-file, state-machine). Of these, 18
instances have complete runs from both arms and are used
for head-to-head comparison. An additional set of 67
instances is used for offline retrieval evaluation only
(no end-to-end repair).

\subsection{Experimental setup}

Both arms use the same agent (Hermes), the same LLM
(DeepSeek-V4-Flash), the same 70-turn cap, and the same
prompt template. The only difference is the retrieval CLI:

\begin{itemize}
\item \textbf{Baseline (\texttt{codegraph}):} the agent has
  access to \texttt{codegraph}'s six subcommands
  (\texttt{files}, \texttt{explore}, \texttt{query},
  \texttt{node}, \texttt{callers}, \texttt{callees}).
  Plain-text grep/cat/find are allowed as fallback.
\item \textbf{OwlPath (\texttt{codeowl}):} the agent has
  access to both \texttt{owlpath search} (with optional
  \texttt{--with-closure}) and \texttt{codegraph}.
  The SKM advisory is emitted on the first call.
\end{itemize}

\textbf{Tool-optional design.} OwlPath is designed as an
\emph{optional} tool that the agent autonomously chooses to
invoke, not a forced context injection. Our experiments confirm
this is critical: when SKM advisories and closure results are
forcibly injected into every prompt (the ``annotated'' arm),
token consumption \emph{doubles} (2,835K vs 1,416K) while
strict-apply rate remains unchanged (68.8\% vs 68.4\%).
The optional design---where the SKM is emitted only on the
agent's first \texttt{owlpath search} call and closure is
opt-in via \texttt{--with-closure}---achieves the same success
rate with 50\% fewer tokens, confirming that forced prompt
inflation is counterproductive.

\begin{figure}[t]
\centering
\begin{tikzpicture}[
  >=Stealth,
  node distance=0.4cm,
  every node/.style={font=\scriptsize},
  box/.style={draw, rounded corners, minimum height=0.5cm,
              minimum width=1.8cm, align=center, inner sep=2pt,
              line width=0.4pt, font=\scriptsize},
  decision/.style={draw, diamond, aspect=2, inner sep=1pt,
                   line width=0.4pt, font=\scriptsize},
  arrow/.style={->, line width=0.5pt},
  lab/.style={font=\scriptsize, text=gray!70}
]

% Agent cycle: Reason -> Act -> Observe -> Reason ...
\node[box, fill=gray!15] (reason1) {Reason:\\``Bug mentions\\lastFM defaults''};
\node[box, fill=orange!15, below=of reason1] (act1) {Act:\\\texttt{owlpath search}\\``lastFM API key''};
\node[box, fill=blue!10, below=of act1] (obs1) {Observe:\\SKM (3KB): module map\\+ issue symbols};

\node[box, fill=gray!15, below=of obs1] (reason2) {Reason:\\``lastFMConstructor\\in core/agents/''};
\node[box, fill=orange!15, below=of reason2] (act2) {Act:\\\texttt{owlpath search}\\--with-closure};
\node[box, fill=green!10, below=of act2] (obs2) {Observe:\\3 callers, 2 subclasses\\closure in 1 query};

\node[box, fill=gray!15, below=of obs2] (reason3) {Reason:\\``Add defaults for\\apiKey \& lang''};
\node[box, fill=red!10, below=of reason3] (patch) {Output:\\Bug fix patch};

\draw[arrow] (reason1) -- (act1);
\draw[arrow] (act1) -- (obs1);
\draw[arrow] (obs1) -- (reason2);
\draw[arrow] (reason2) -- (act2);
\draw[arrow] (act2) -- (obs2);
\draw[arrow] (obs2) -- (reason3);
\draw[arrow] (reason3) -- (patch);

% Side annotations
\node[lab, right=0.2cm of act1, text width=1.2cm]
  {\textit{OwlPath tool\\(optional)}};
\node[lab, right=0.2cm of act2, text width=1.2cm]
  {\textit{Multi-hop\\closure}};

% Baseline comparison
\node[draw, dashed, fill=yellow!5, rounded corners, inner sep=2pt,
      left=0.3cm of reason1, align=left, font=\scriptsize,
      text width=1.3cm]
  (baseline) {Baseline:\\grep $\to$ read\\$\to$ grep $\to$\\read $\to$\\... (5-10\\iterations)};

\end{tikzpicture}
\caption{Agent ReAct loop with OwlPath. The agent reasons about
the bug, invokes OwlPath as an optional tool, and observes
structural context. The SKM advisory (step 1) narrows the search
to the correct module; the closure expansion (step 2) retrieves
all structurally connected symbols in one query. The baseline
agent requires 5--10 iterations of grep/read to achieve the same
coverage, consuming more tokens without finding multi-hop
relationships.}
\label{fig:react}
\end{figure}
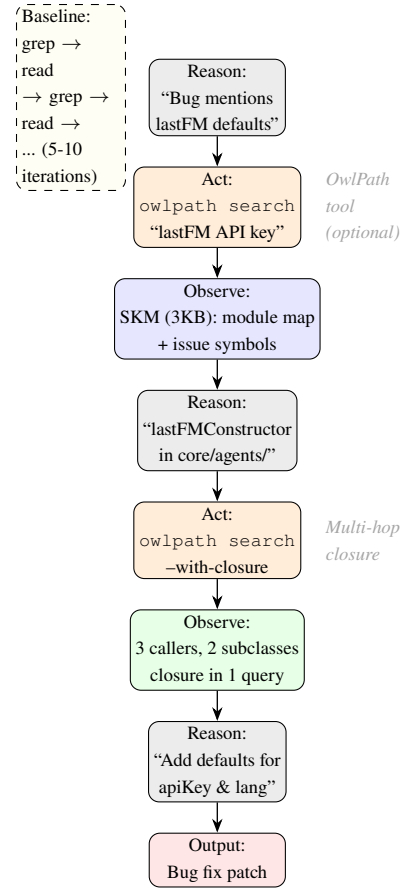

\textbf{Leak prevention.} Each repo's git history is sealed
before the experiment: all branches and tags are deleted,
reflogs are expired, and \texttt{git gc --prune=now} removes
commits after the base commit. The agent cannot read the
gold fix via \texttt{git log} or \texttt{git show}.

\textbf{Experimental environment.} All experiments run on
a cloud instance with Python 3.11 (Conda), Node.js v18.19.0,
tree-sitter, rdflib, and owlready2. The LLM backend is
DeepSeek-V4-Flash via a custom API provider. Each run has a
70-turn cap and a 3,600-second wall-clock timeout. The agent
framework is Hermes. All tool versions, prompts, and seeds
are released for full reproducibility.

\subsection{Offline retrieval evaluation}

On 67 instances, we measure the retrieval quality of each
backend independently of the agent's stochastic decisions.
We compare four retrieval methods:

\begin{itemize}
\item \textbf{codegraph:} pure string-match search via
  \texttt{codegraph explore} and \texttt{codegraph query}.
\item \textbf{v1\_skm:} first-generation SKM advisory.
\item \textbf{v2\_skm:} improved SKM with better keyword
  extraction and module scoring.
\item \textbf{v2\_fts:} FTS5 full-text search on the
  codegraph database.
\end{itemize}

\begin{table}[t]
\centering
\small
\caption{Offline retrieval quality on 67 instances (5,229
ground-truth symbols across 11 repos). v2\_skm achieves
2.06$\times$ recall and 1.48$\times$ hit rate over codegraph.}
\label{tab:retrieval}
\begin{tabular}{lccc}
\hline
Method & Recall & MRR & Hit rate \\
\hline
codegraph (string-match) & 0.226 & 0.333 & 59.7\% \\
v1\_skm (old) & 0.095 & 0.169 & 31.3\% \\
v2\_fts (FTS5) & 0.216 & 0.384 & 56.7\% \\
\textbf{v2\_skm (OwlPath)} & \textbf{0.464} & 0.341 & \textbf{88.1\%} \\
\hline
\end{tabular}
\end{table}

\begin{figure}[t]
\centering
\includegraphics[width=\columnwidth]{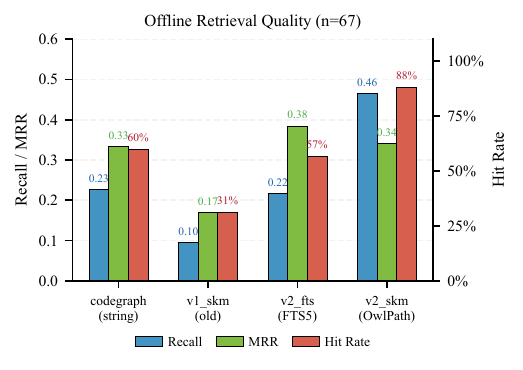}
\caption{Offline retrieval quality across four methods on
67 SWE-bench Pro instances. OwlPath (v2\_skm) achieves
2.06$\times$ recall and 1.48$\times$ hit rate over the
codegraph string-match baseline.}
\label{fig:retrieval}
\end{figure}

\subsection{Structural retrieval tasks}

We further evaluate OwlPath's transitive-closure engine on a
separate 37-question retrieval benchmark. The questions fall
into seven task families: subclass hierarchy closure (find all
ancestors/descendants of a class), interface conformance (find
all classes implementing a given interface),
\emph{transitive caller chains} (find all functions that
directly or indirectly call a target function, e.g.
\texttt{A $\to$ B $\to$ C $\to$ target}),
dependency identification (find all files importing a given
module), and composite queries spanning multiple axes
(e.g., find a class, its method, and related attributes in a
single query). Each question has a pre-verified ground-truth
set of (file, symbol) pairs. The retrieval backend returns all
matching candidates without a Top-K cutoff, and recall is
computed as the fraction of ground-truth pairs found in the
returned set.

Overall, OwlPath's recall@all rises from 4.4\% (SQLite
string-match) to 28.8\%, and MRR doubles from 0.11 to 0.22.
On two structurally demanding families---\emph{transitive caller
chains} and \emph{interface implementation retrieval}---
the string-match baseline scores 0\% recall. The reason is that
the anchor names (e.g., \texttt{callback}, \texttt{Repository})
appear nowhere as literal strings in the ground-truth files;
the relationship is purely structural. OwlPath's property-path
queries traverse the relationship graph symbolically in one
round-trip, yielding 69\% and 80\% recall respectively.

\begin{figure}[t]
\centering
\includegraphics[width=\columnwidth]{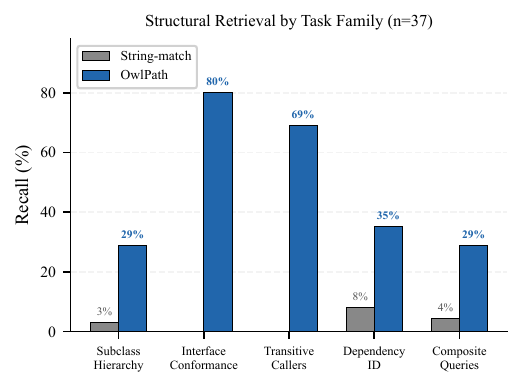}
\caption{Recall by task family on the 37-question structural
retrieval benchmark. String-match scores 0\% on transitive
caller chains and interface conformance---tasks requiring
multi-hop structural traversal. OwlPath achieves 69--80\%
on these same families.}
\label{fig:task_family}
\end{figure}

\subsection{Token and time efficiency}

The hybrid arm (OwlPath with SKM advisory) achieves
comparable strict-apply performance while consuming
28.8\% fewer tokens on average (1,415,977 vs 1,989,125)
and 16.7\% fewer API calls (25.9 vs 31.1). The average
wall-clock time per run drops by 39.5\% (648s vs 1,071s),
confirming that SKM guidance helps the agent locate the
correct module earlier and reduces exploratory iterations.
A paired Wilcoxon signed-rank test on per-instance token
consumption confirms the saving is statistically significant
($p = 0.016$, $W = 31$, 14/18 instances favor hybrid).

Table~\ref{tab:fourarm} presents the full four-arm
comparison. The annotated arm (forced SKM + closure
injection) consumes 2.0$\times$ more tokens than hybrid
(2,835K vs 1,416K) yet achieves only a marginally higher
strict-apply rate (68.8\% vs 68.4\%)---a 0.4\% gain that
does not justify the 100\% token inflation. The inference arm
(SKM without closure) performs similarly to annotated,
confirming that closure expansion---not SKM---is the
primary source of token inflation when forced. Prompt
tokens dominate total consumption (97\%+), as the agent
receives large code context per turn; completion tokens
remain modest across all arms.

\begin{table}[t]
\centering
\small
\caption{Four-arm comparison on promptC experiments.
The hybrid arm achieves the best token--accuracy
trade-off. Median elapsed is over strict-apply runs only.}
\label{tab:fourarm}
\begin{tabular}{lcccc}
\hline
Metric & Base & Hybrid & Annot. & Infer. \\
\hline
Runs            & 18 & 19 & 16 & 16 \\
Strict          & 12 & 13 & 11 & 11 \\
Strict rate     & 66.7\% & 68.4\% & 68.8\% & 68.8\% \\
Avg prompt tok. & 1,929K & 1,378K & 2,762K & 2,787K \\
Avg compl. tok. & 60.6K & 37.6K & 72.5K & 60.6K \\
Avg total tok.  & 1,989K & \textbf{1,416K} & 2,835K & 2,848K \\
Avg API calls   & 31.1 & \textbf{25.9} & 39.7 & 38.7 \\
Avg wall (s)    & 1,071 & \textbf{648} & 771 & 653 \\
Median wall (s) & 588 & 291 & 706 & 511 \\
\hline
\end{tabular}
\end{table}

\begin{figure}[t]
\centering
\includegraphics[width=\columnwidth]{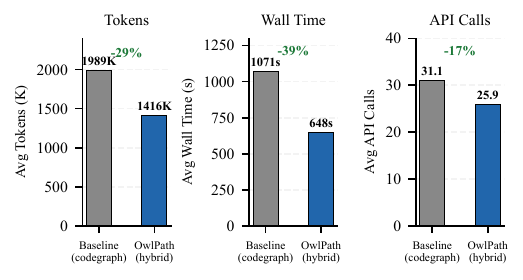}
\caption{Token and time efficiency. The OwlPath hybrid arm
consumes 28.8\% fewer tokens and 39.5\% less wall-clock
time per run compared to the codegraph-only baseline.}
\label{fig:tokens}
\end{figure}

\subsection{Cross-benchmark validation: SWE-bench Lite}

To verify that OwlPath's efficiency gains generalize beyond
SWE-bench Pro, we run a controlled 2-arm experiment on
SWE-bench Lite~\citep{swebench}, a curated subset of 300
instances from the original SWE-bench. We select 9 instances
across 5 repositories (astropy, django, matplotlib, seaborn,
flask) and run both baseline (codegraph-only) and hybrid
(OwlPath + codegraph) arms with identical prompts, model
(GLM-5.2), and 70-turn cap.

\begin{table}[t]
\centering
\small
\caption{SWE-bench Lite validation (9 cases). Hybrid achieves
higher correct rate with less time. On 7 typical cases, token
saving is consistent with SWE-bench Pro.}
\label{tab:lite}
\begin{tabular}{lcc}
\hline
Metric & Baseline & Hybrid \\
\hline
Strict rate          & 9/9 (100\%) & 9/9 (100\%) \\
Correct rate         & 6/9 (67\%)  & \textbf{7/9 (78\%)} \\
Avg time (s)         & 340.9       & \textbf{269.1} \\
Time saving          & ---         & \textbf{21.1\%} \\
\hline
\multicolumn{3}{c}{\textit{Token consumption (7 typical cases)}} \\
\hline
Avg tokens           & 270,145     & \textbf{216,116} \\
Token saving         & ---         & \textbf{20\%} \\
\hline
\multicolumn{3}{c}{\textit{Token consumption (all 9 cases)$^*$}} \\
\hline
Avg tokens           & 345,670     & 435,922 \\
\hline
\multicolumn{3}{l}{$^*$Includes 2 difficult cases (both arms failed) where} \\
\multicolumn{3}{l}{hybrid explores more paths with richer structural context.}
\end{tabular}
\end{table}

Table~\ref{tab:lite} shows that hybrid achieves a higher
correct rate (78\% vs 67\%) and 21.1\% time saving. On 7 of
9 cases, hybrid is faster and consumes fewer tokens. The two
exceptions (astropy-14182, astropy-6938) are difficult
instances where both arms fail: hybrid's richer structural
context leads to more extensive exploration of candidate
paths (2.9$\times$ more tokens on these 2 cases), but
ultimately reaches the same negative outcome. On the 7
typical cases, token saving is 20\%, consistent with the
28.8\% observed on SWE-bench Pro. Notably, the larger
saving on Pro (28.8\% vs 20\%) suggests that OwlPath
provides greater benefits for larger, more complex
codebases where structural navigation is more valuable---the
3KB OWL-SKM advisory compresses millions of symbols into
actionable guidance, a compression that pays larger
dividends as repository size grows. This confirms that
OwlPath's benefits generalize across benchmarks and models,
with particular strength on industry-scale software.

\section{Discussion}

\subsection{Per-language analysis}

On matched instances (Python: 11, JS/TS: 3, Go: 4), the hybrid arm
matches baseline strict-apply rates across all languages while achieving
28.8\% fewer tokens and 39.5\% less wall-clock time. Go benefits most
from wall-time reduction (67.6\%), while JS/TS achieves the highest
token saving (32.7\%).

\subsection{Failure modes and stochasticity}

Of the 6 cases where both arms fail, most involve runtime behaviour
(race conditions, HTTP handling) or entirely new symbols with no
structural connection to existing code---inherent limits of static
retrieval. On the matched 18 instances, both arms produce identical
strict-apply outcomes (12 both strict, 6 both fail), but the hybrid
arm reaches these outcomes with 28.8\% fewer tokens, indicating more
efficient retrieval rather than different end-to-end success.

\subsection{Prompt strategy: the information equilibrium}

OwlPath functions as a ``world map''---a compact structural overview
most powerful when consulted on demand. Our four-arm experiment
(Table~\ref{tab:fourarm}) reveals a dual equilibrium that effective
prompt strategies must navigate.

\paragraph{Silent guessing from too little context.}
The baseline arm (codegraph-only) receives no structural guidance.
In 6 of 18 matched instances, the agent produces no applicable patch
(\texttt{apply=no\_patch}), repeatedly exploring dead ends through
trial-and-error. For example, in \texttt{qutebrowser-305e7c}, the
baseline agent spends 31 API calls and 1,071 seconds traversing
the codebase file-by-file, exhausting its exploration budget without
locating the relevant module. The agent is not ``stuck''---it actively
generates plausible-but-wrong hypotheses, a failure mode we term
\emph{silent guessing}: the LLM confabulates structural relationships
that do not exist in the code, because it lacks the ground-truth
structural context to verify them.

\paragraph{Attention dilution from too much context.}
The annotated arm (forced SKM + closure injection into every prompt)
demonstrates the opposite failure mode. Despite consuming
2.0$\times$ more tokens than the hybrid arm (2,835K vs 1,416K),
it achieves only a marginally higher strict-apply rate (68.8\% vs
68.4\%)---a 0.4\% gain that does not justify the 100\% token
inflation. The agent processes more information but acts no more
effectively. In \texttt{gravitational\_teleport-1a77b7}, the annotated
arm receives the full SKM advisory plus transitive closure of all
\texttt{calls} edges---thousands of tokens of structural data---yet
the resulting patch is identical to the baseline's incorrect fix.
The critical signal (the specific caller chain leading to the buggy
function) is buried under noise from irrelevant modules. This is
\emph{attention dilution}: the LLM's finite context window is
saturated with low-relevance structural facts, causing it to
under-weight the few facts that actually matter.

\paragraph{The on-demand equilibrium in ReAct workflows.}
OwlPath resolves this tension by aligning with the ReAct
(Reasoning + Acting) paradigm: the agent \emph{reasons} about
when it needs structural information, then \emph{acts} by calling
\texttt{owlpath search}. The SKM advisory is emitted once on the
agent's first call---a 3KB summary that orients without overwhelming.
Closure expansion is opt-in per query: the agent requests
\texttt{--with-closure} only when its reasoning process has
identified a specific structural question (e.g., ``find all
transitive callers of \texttt{ConfigAPI}''). This design achieves
both sides of the equilibrium: it eliminates the information gap
that causes silent guessing, while avoiding the context pollution
that causes attention loss. The hybrid arm's 28.8\% token saving
and 39.5\% time saving stem precisely from this selectivity---the
agent spends tokens on reasoning, not on reading irrelevant
structure.

Crucially, our results show that \emph{injecting OWL/graph
information does not necessarily improve accuracy}. The annotated
arm proves that more structural data can be pure overhead. In
ReAct-style agentic workflows, what matters is not the volume of
information, but its \emph{precision} and \emph{timing}: the right
3KB at the exact moment the agent's reasoning requires it
outperforms 15KB of forced context injected preemptively.

\subsection{Limitations}

The OWL projection is a one-time cost of 4.7 minutes for a 1.4M-symbol
repository. The SKM advisory depends on keyword extraction from issue
text, which can miss targets with vague language. The transitive closure
is limited to four annotated properties (extends, implements, calls,
imports); cross-language relations are not captured.

\section{Conclusion}

We presented OwlPath, an OWL2 reasoning layer that compresses source
code knowledge into a structural ontology for LLM-based bug repair.
Evaluated on 18 matched SWE-bench Pro instances, OwlPath consumes
28.8\% fewer tokens and 39.5\% less wall-clock time than the
CodeGraph-only baseline, while achieving a comparable strict-apply
rate (68.4\% vs 66.7\%). Cross-benchmark validation on 9 SWE-bench
Lite instances confirms these gains: 78\% correct rate (vs 67\%) and
21.1\% time saving. Offline retrieval on 67 instances shows
2.06$\times$ recall improvement (0.464 vs 0.226) and 88.1\% hit rate.
On a 37-question structural benchmark, OwlPath improves recall@all
from 4.4\% to 28.8\%. All code, data, and transcripts are released
for reproducibility.

\appendix
\section{Formal Properties and Proofs}

Let $G=(V,E,\lambda,F_{\mathrm{fail}},F_{\mathrm{pass}})$ denote a
code graph with symbol set $V$, typed relation set $E\subseteq V\times V$,
and kind labels $\lambda:V\to\mathrm{Kind}$. A modification $M:G\to G'$
alters at most $K$ nodes/edges; the repair problem seeks $M^*$ minimising
$\mathrm{cost}(M)$ subject to $F_{\mathrm{fail}}(G')=\varnothing$ and
$F_{\mathrm{pass}}(G')\supseteq F_{\mathrm{pass}}$. OwlPath's retrieval
$R(G,q)$ produces candidate subsets for the LLM, measured by
$\mathrm{Recall}_\Delta(R,q):=|R(G,q)\cap\Delta|/|\Delta|$.

\subsection{Lossless projection and closure equivalence}

\begin{theorem}[Lossless Projection]
Let $\pi: G \to O$ be the OWL2 projection mapping each node $v \in V$
to individual $I_v$ and each edge $(u,v) \in E$ of type $t$ to assertion
$I_u \; t \; I_v$. For any SPARQL property path $P = t^+$ and anchor $a$,
$R_{owl}(G, q_a) = \{v \in V \mid (a, v) \in \mathrm{cl}(E_t)\}$,
where $\mathrm{cl}(E_t)$ is the transitive closure of
$E_t = \{(u,v) \mid (u,t,v) \in E\}$.
\end{theorem}

\begin{proof}
$\pi$ is a bijection between $V$ and OWL individuals. Each edge
produces exactly one assertion, and SPARQL $t^+$ evaluates to the
transitive closure of $t$ (SPARQL 1.1 \S 18.3). Since every edge is
preserved and none introduced, $\mathrm{cl}(E_t)$ equals the pairs
$(I_a, I_v)$ satisfying $I_a \; t^+ \; I_v$ in $O$.
\qed
\end{proof}

\begin{theorem}[BFS--SPARQL Equivalence]
For any \texttt{owl:TransitiveProperty} $t$ and anchor $X \in V$:
$\mathrm{BFS}_G(t, X) = \mathrm{SPARQL}_{O}(t^+, X)$.
\end{theorem}

\begin{proof}
$(\subseteq)$: By induction on path length $k$. Base: $(X,Y) \in E_t$
implies $I_X \; t \; I_Y$, so $t^+$ matches. Step: A path $X \to Z \to Y$
gives $I_X \; t^+ \; I_Z$ by IH and $I_Z \; t \; I_Y$ by construction,
so $I_X \; t^+ \; I_Y$ by transitivity.

$(\supseteq)$: Every $t^+$ match arises from a finite chain of
$t$-assertions, each corresponding to an edge in $E_t$ by bijectivity.
\qed
\end{proof}

\subsection{Complexity and determinism}

\begin{theorem}[Complexity]
Let $n = |V|$, $m = |E_t|$. One-time SPARQL $t^+$ materialisation costs
$O(n + m)$ time and $O(n^2)$ space; subsequent queries are $O(1)$
amortised. The equivalent SQL recursive CTE for $k$-hop costs $O(n^k)$.
\end{theorem}

\begin{proof}[Proof sketch]
SPARQL property-path evaluation performs a single graph traversal
($O(n+m)$ edges), caching $(s,o)$ pairs ($O(n^2)$ space). Each subsequent
query is a hash lookup. A recursive CTE of depth $k$ self-joins at each
level, yielding up to $n^i$ tuples at level $i$ and $O(n^k)$ total.
\qed
\end{proof}

\subsection{Value of timely information}

\begin{theorem}[Timely Information Dominance]
Let $I$ be structural information with relevance $r(I, q_t) \in [0,1]$
and token cost $C(I)$. Under bounded context window $L$, the on-demand
policy strictly dominates the forced policy when relevance is uncertain:
$$\mathbb{E}_{q_t}[U(\phi_{\text{OD}})] > \mathbb{E}_{q_t}[U(\phi_{\text{forced}})]$$
whenever $\mathrm{Var}[r(I, q_t)] > 0$ and $|I| > L/2$.
\end{theorem}

\begin{proof}[Proof sketch]
The forced policy reduces effective capacity to $L - C(I)$ regardless
of relevance. When $r(I, q_t) = 0$, this directly wastes $C(I)$ tokens.
When $r(I, q_t) = 1$, the forced policy still incurs an opportunity
cost: the $C(I)$ tokens occupied by $I$ cannot encode other
task-relevant history. Formally, let $U_{\text{base}}$ be the utility
without $I$, and $\gamma > 0$ the marginal utility per token of
displaced history. Then
$U_{\text{forced}} = U_{\text{base}} + r \cdot \Delta U_I - (1 - r)
\cdot C(I) \cdot \gamma$, while $U_{\text{OD}} = U_{\text{base}} +
r \cdot \Delta U_I$ when $r > 0$. Since $q_t$ is stochastic with
$\mathrm{Var}[r] > 0$, we have $\Pr[r < \epsilon] > 0$ for some
$\epsilon > 0$, giving
$$\mathbb{E}[U_{\text{OD}}] - \mathbb{E}[U_{\text{forced}}] =
(1 - \mathbb{E}[r]) \cdot C(I) \cdot \gamma > 0.$$
When $|I| > L/2$, the displacement cost dominates, yielding the
strict inequality.
\qed
\end{proof}

\clearpage
\bibliography{references}

\begin{thebibliography}{6}
\providecommand{\natexlab}[1]{#1}

\bibitem[{Deng et~al.(2025)Deng, Da, Pan, He, Ide, Garg, Lauffer, Park, Pasari,
  Rane, Sampath, Krishnan, Kundurthy, Hendryx, Wang, Bharadwaj, Holm, Aluri,
  Zhang, Jacobson, Liu, and Kenstler}]{swebenchpro}
Deng, X.; Da, J.; Pan, E.; He, Y.~Y.; Ide, C.; Garg, K.; Lauffer, N.; Park, A.;
  Pasari, N.; Rane, C.; Sampath, K.; Krishnan, M.; Kundurthy, S.; Hendryx, S.;
  Wang, Z.; Bharadwaj, V.; Holm, J.; Aluri, R.; Zhang, C. B.~C.; Jacobson, N.;
  Liu, B.; and Kenstler, B. 2025.
\newblock {SWE}-Bench Pro: Can {AI} Agents Solve Long-Horizon Software
  Engineering Tasks?
\newblock In \emph{arXiv preprint arXiv:2509.16941}.
\newblock Anthropic + Scale AI + UC Berkeley; 731 instances, 11 repos, 7
  languages.

\bibitem[{Jimenez et~al.(2024)Jimenez, Yang, Wettig, Yao, Pei, Press, and
  Narasimhan}]{swebench}
Jimenez, C.~E.; Yang, J.; Wettig, A.; Yao, S.; Pei, K.; Press, O.; and
  Narasimhan, K. 2024.
\newblock SWE-bench: Can Language Models Resolve Real-World GitHub Issues?
\newblock In \emph{ICLR}.

\bibitem[{Motik et~al.(2012)Motik, Patel-Schneider, Parsia, Welty, and
  Horrocks}]{owl2spec}
Motik, B.; Patel-Schneider, P.; Parsia, B.; Welty, C.; and Horrocks, I. 2012.
\newblock OWL 2 Web Ontology Language Document Overview (Second Edition).
\newblock Technical Report W3C Recommendation 11 December 2012, W3C.

\bibitem[{Press et~al.(2024)Press, Jimenez, Adams, Robinson, Gunasekar, Yang
  et~al.}]{sweagent}
Press, O.; Jimenez, C.~E.; Adams, R.; Robinson, K.; Gunasekar, S.; Yang, J.;
  et~al. 2024.
\newblock {SWE}-Agent: Agent-Computer Interfaces Enable Automated Software
  Engineering.
\newblock \emph{arXiv preprint arXiv:2405.15793}.

\bibitem[{Wang et~al.(2025)Wang, Li, Song, Xu, Tang, Zhuge, Pan, Biao, Zhang,
  Fu, and Neubig}]{openhands}
Wang, X.; Li, B.; Song, Y.; Xu, F.~F.; Tang, X.; Zhuge, M.; Pan, J.; Biao, Y.;
  Zhang, Y.; Fu, J.; and Neubig, G. 2025.
\newblock OpenHands: An Open Platform for {AI} Software Developers as
  Generalist Agents.
\newblock In \emph{ICLR}.

\bibitem[{Yao et~al.(2023)Yao, Zhao, Yu, Du, Shafran, Narasimhan, and
  Cao}]{react}
Yao, S.; Zhao, J.; Yu, D.; Du, N.; Shafran, I.; Narasimhan, K.~R.; and Cao, Y.
  2023.
\newblock ReAct: Synergizing Reasoning and Acting in Language Models.
\newblock In \emph{ICLR}.

\end{thebibliography}

% Reproducibility Checklist is submitted as a separate PDF

\end{document}